\documentclass[12pt]{article}
\usepackage{amsmath, amssymb, amsthm, verbatim}

\newtheorem{theorem}{Theorem}[section]

\newtheorem{corollary}[theorem]{Corollary}

\newcommand{\IR}{\mathbb{R}}
\newcommand{\IC}{\mathbb{C}}

\newcommand{\row}{\text{rowspace }}
\newcommand{\diag}{\text{diag}}

\def\cS{{\cal S}}

\def\diag{{\rm diag}\,}
\def\tr{{\rm tr}\,}

\begin{document}
\openup .95 \jot
\title{Physical transformations between quantum states}
\author{Zejun Huang, Chi-Kwong Li, Edward Poon, Nung-Sing Sze}
\date{}
\maketitle

\begin{abstract}
Given two sets of quantum states
$\{A_1, \dots, A_k\}$ and $\{B_1, \dots, B_k\}$, represented as sets as density matrices,
necessary and sufficient conditions are obtained for the
existence of a physical transformation $T$, represented as a trace-preserving
completely positive map, such that
$T(A_i) = B_i$ for $i = 1, \dots, k$.
General completely positive maps without the trace-preserving requirement,
and unital completely positive maps transforming the states are also considered.
\end{abstract}


\section{Introduction and Notation}

\subsection{Introduction}

In quantum information science, quantum states with $n$ physically measurable states
are represented by $n\times n$ \emph{density matrices}, i.e., positive semidefinite matrices
with trace one. In particular, \emph{pure} states are rank one
density matrices, while \emph{mixed} states have rank greater than one.
We are interested in studying the conditions on two sets of quantum states
$\{A_1, \dots, A_k\}$ and $\{B_1, \dots, B_k\}$ so that there is a physical
transformation (a.k.a. quantum operation or quantum channel) $T$ such that
$T$ sends $A_i$ to $B_i$ for $i = 1, \dots, k$.

To set up the mathematical framework,
let $M_{m,n}$ be the set of $m\times n$ complex matrices, and use
the abbreviation $M_n$ for $M_{n,n}$.
Denote by $x^*$ and $A^*$ the conjugate transpose of  vectors $x$ and
matrices $A$.
Physical transformations sending quantum states (represented as density matrices)
in $M_n$ to quantum states in $M_m$ are trace-preserving completely positive (TPCP)
maps $T: M_n \rightarrow M_m$ with an operator sum representation
\begin{equation}\label{E:cp}
T(X) = \sum_{j=1}^r F_j X F_j^*,
\end{equation}
where $F_1, \dots, F_r$ are $m\times n$ matrices satisfying $\sum_{j=1}^r F_j^*F_j = I_n$;
see \cite{C}, \cite{K}, \cite[\S8.2.3]{NC}.
So, we are interested in studying the conditions for the existence of
a TPCP map $T$ of the form (\ref{E:cp}) with  $\sum_{j=1}^r F_j^*F_j = I_n$
such that $T(A_i) = B_i$ for $i = 1, \dots, k$.

We also consider more general types of physical transformations
(completely positive (CP) linear maps) without the trace-preserving assumption,
i.e., not requiring $\sum_{j=1}^r F_j^*F_j = I_n$. Such operations are also considered
in the study of quantum information science; see \cite[\S8.2.4]{NC}.
Furthermore, in Section 4 we consider \emph{unital} completely positive maps
which are of interest in the theory of $C^*$-algebras.  Such CP maps are dual to the trace-preserving ones and send the identity matrix to the identity matrix, i.e., they satisfy $\sum_{j=1}^r F_j F_j^* = I_m$. 

In Section 2, we study physical transformations on qubit states, i.e., quantum states on $M_2$.
Section 3 concerns physical transformations sending general states to general states,
and Section 4 concerns more general transformations acting on pure states.

\subsection{Notation}

We conclude this section by defining additional notation and recalling some terminology that will be used later.
Given a matrix $M$ (which we may alternatively denote as $(M_{ij})$, to focus on its entries), we write $M^t$ for the transpose of $M$, and $\bar{M}$ for the matrix whose $(i,j)$-entry is the complex conjugate of $M_{ij}$.
The \emph{Hadamard product} (or \emph{Schur product}) of two $m \times n$ matrices $A$ and $B$ is the $m \times n$ matrix $A \circ B$ whose $(i,j)$-entry is given by $A_{ij} B_{ij}$.  (So, the $\circ$ symbol denotes entry-wise multiplication.)
A \emph{correlation} matrix is a positive semidefinite matrix with all diagonal entries equal to 1.

Suppose a matrix $A$ has the spectral decomposition $A = \sum_{k=1}^m \lambda_k v_k v_k^*$ for some orthonormal eigenvectors $v_k$.  One possible \emph{purification} for $A$ is the vector $\sum_{k=1}^m \sqrt{\lambda_k} v_k \otimes v_k$; the most general form for purifications of $A$ are vectors of the form $\phi = \sum_{k=1}^m \sqrt{\lambda_k} v_k \otimes W v_k \in \IC^m \otimes \IC^r$, where $W$ is a partial isometry from $\IC^m$ to $\IC^r$.
Note that, for any purification $\phi$ of $A$, the partial trace of $\phi \phi^*$ over the second system is precisely $A$, and one can actually take a more abstract point of view and define purifications to be those vectors possessing this property.  (Recall that the partial trace of $B \otimes C \in M_m \otimes M_r$ over the second system is just $B (\tr C)$, and one extends linearly to define the partial trace on all of $M_m \otimes M_r$.)

\section{Qubit states}
\medskip

In this section we focus solely on qubit states ($2 \times 2$ density matrices). Recall that the trace norm $\| \cdot \|_1$ of a matrix $X$ is the sum of its singular values.
The following interesting result was proved in \cite{AU}; see also \cite{CJW}.

\begin{theorem} \label{2.1}
Let $A_1, A_2, B_1, B_2 \in M_2$ be density matrices.
There is a TPCP map sending $A_i$ to $B_i$ for $i = 1, 2$ if and only if
$\|A_1-tA_2\|_1  \ge \|B_1 - tB_2\|_1$ for all $t \ge 0$.
\end{theorem}

The proof in \cite{AU}  is quite long. In the following we give a short proof of the
result, and give another condition that is much easier to check (condition (c) in
Theorem \ref{T:qubit}) by making the following reduction: if rank $A_1 = 2$,
then we can find $c > 0$ so that $\tilde{A_1} = A_1 - cA_2$ is a positive semidefinite
matrix of rank one.  Then we simply replace $A_1, B_1$ by $\tilde{A_1}$, $\tilde{B_1}
= B_1 - c B_2$, since a TPCP map sending $A_i$ to $B_i$ exists if and only if there is
a TPCP map sending $\tilde{A_1}$ to $\tilde{B_1}$ and $A_2$ to $B_2$.
We may then repeat the process by considering $\tilde{A_2} = A_2 - \tilde c \tilde{A_1}$.

So, by taking linear combinations of $A_1, A_2$ (and the corresponding combinations of $B_1$, $B_2$),
we may assume that $A_1 = x_1x_1^*$ and $A_2 = x_2x_2^*$.
We have the following.

\begin{theorem} \label{T:qubit}  Let $A_1 = x_1x_1^*, A_2 = x_2x_2^*, B_1, B_2 \in M_2$
be density matrices. The following conditions are equivalent.

\medskip
{\rm (a)} There is a TPCP map sending $A_i$ to $B_i$ for $i = 1, 2$.

{\rm (b)} $\sqrt{(1+t)^2 - 4t|x_1^*x_2|^2 }=
\|A_1-tA_2\|_1  \ge \|B_1 - tB_2\|_1$ for all $t \ge 0$.

{\rm (c)} $|x_1^*x_2| = \|\sqrt{A_1}\sqrt{A_2}\|_1 \le \|\sqrt{B_1}\sqrt{B_2}\|_1$.
\end{theorem}

Note that condition (c) is of independent interest, for it relates the fidelity between the initial states with the fidelity $\| \sqrt{B_1} \sqrt{B_2} \|_1$ between the final states $B_1$, $B_2$, and can be generalized to give a necessary (but not sufficient) condition for the existence of a TPCP map sending $k$ initial states to $k$ final states (see equation \eqref{E:fidelity} later, also \cite{CJW}).

\it Proof. \rm Note that for $X \in M_2$, $\|X\|_1^2 = \tr(XX^*) + 2|\det(X)|$.
One can readily verify the first equality in (b) and the first equality in (c).

\medskip
(a) $\Rightarrow$ (b). Suppose $T$ is TPCP. If $A = A_+ - A_-$ where $A_+$ and $A_-$ are positive semidefinite,
then
$$\|T(A)\|_1 \le \|T(A_+)\|_1 + \|T(A_-)\|_1 =
\tr T(A_+) + \tr T(A_-) = \tr A_+ + \tr A_- = \|A\|_1.$$
Thus $\|B_1-tB_2\|_1 = \|T(A_1-tA_2)\|_1 \le  \|A_1 - tA_2\|_1$ for all $t \ge 0$.

\medskip
(b) $\Rightarrow$ (c).
Suppose one of the matrices $B_1$ and $B_2$  has rank 1.
Without loss of generality, we may assume that $B_2 = y_2 y_2^*$.
By condition (b), for $t > y_2^*B_1y_2$, we have
\begin{eqnarray*}
(1+t)^2 - 4t|x_1^*x_2|^2
&\ge& \|B_1 - t y_2y_2^*\|_1^2 \cr
&=&  \tr((B_1 - ty_2y_2^*)^2) + 2 |\det(B_1 - ty_2y_2^*)| \cr
&=&  t^2 + 2t - 4 t(y_2^*B_1y_2) + \gamma
\end{eqnarray*}
for a constant $\gamma \in \IR$.
Thus, $|x_1^*x_2|^2 \le y_2^*B_1y_2 = \|\sqrt{B_1}\sqrt{B_2}\|_1^2$.

\medskip
Suppose both $B_1$ and $B_2$ are invertible.
Choose $t$ so that $\det(B_1) = \det(tB_2)$.
Applying a suitable unitary similarity transform, we may assume that
$B_1-tB_2$ is in diagonal form so that
$$B_1 = \begin{bmatrix} b_1  & c \cr \bar c & 1-b_1 \end{bmatrix},
\ B_2 = \begin{bmatrix} b_2  & c/t \cr \bar c/t & 1-b_2 \end{bmatrix}.$$
Then
\begin{eqnarray*}
&&\det(B_1+tB_2) - |\det(B_1-tB_2)|\cr
&=& [(b_1+tb_2)(1+t - b_1-tb_2) - 4|c|^2] - (b_1-tb_2)((b_1-tb_2)-(1-t)) \cr
&=& 2\{b_1(1-b_1)-|c|^2) + t^2(b_2(1-b_2) - |c|^2/t^2)\} \cr
&=& 2 (\det(B_1) + \det(tB_2)) \cr
&=& 4 \det(\sqrt{B_1}\sqrt{tB_2})  \qquad \mbox{ because $t$ satisfies } \ \det(B_1) = \det(tB_2)\cr
&=& 4t \det(\sqrt{B_1}\sqrt{B_2}).
\end{eqnarray*}
Hence,
\begin{equation} \label{eq1}
\det(B_1+tB_2) - |\det(B_1-tB_2)| = 4 t \det(\sqrt{B_1}\sqrt{B_2}).
\end{equation}
By condition (b), we have
\begin{eqnarray*}
&&(1+t)^2 - 4t|x_1^*x_2|^2\cr
&\ge& \tr((B_1 - tB_2)^2) + 2|\det(B_1-tB_2)|\cr
& = &  \tr((B_1 + tB_2)^2)- 2t\tr(B_1B_2+B_2B_1) \cr
&& \quad  + 2\det(B_1+tB_2) - 2\det(B_1+tB_2) + 2|\det(B_1-tB_2)| \cr
& = & (\tr(B_1+tB_2))^2 - 4t\tr(B_1B_2) - 2\det(B_1+tB_2)+2|\det(B_1-tB_2)|\cr
& = & (1+t)^2 - 4t\left[\tr(B_1B_2) + 2\det(\sqrt{B_1}\sqrt{B_2})\right]
\qquad \mbox{ by (\ref{eq1}) } \cr
& = & (1+t)^2 - 4t\|\sqrt{B_1}\sqrt{B_2}\|_1^2.
\end{eqnarray*}
Thus, $\|\sqrt{B_1}\sqrt{B_2}\|_1 \ge |x_1^*x_2|$, and condition (c) holds.

\medskip
(c) $\Rightarrow$ (a).  Note that $\|X\|_1 = \max\{|\tr XW|: W \ \hbox{is unitary}\},$ so there exists a unitary $V \in M_2$ such that $|\tr \sqrt{B_1} \sqrt{B_2} V| \geq |x_1^* x_2|$.  If we write $\sqrt{B_1} = [y_1 | y_2]$ and
$\sqrt{B_2}V = [z_1 | z_2]$, and set $y = \begin{bmatrix} y_1 \\ y_2 \end{bmatrix} \in \IC^4$ and
$z = \begin{bmatrix} z_1 \\ z_2 \end{bmatrix} \in \IC^4$, then this inequality implies that $|y^*z| \geq |x_1^* x_2|$.
Set $\delta = 1$ if $y^*z = 0$; otherwise let $\delta = (x_1^* x_2)/(y^*z)$.  Then the $8 \times 2$ matrices
$$X = \begin{bmatrix} x_1 & x_2 \\ 0 & 0 \end{bmatrix} \quad \text{and} \quad Y = \begin{bmatrix} y & \delta z \\ 0 & \sqrt{1 - |\delta|^2} \, z \end{bmatrix}$$
satisfy $X^*X = Y^*Y$ (note that $y_1 y_1^* + y_2 y_2^* = B_1$ and $z_1 z_1^* + z_2 z_2^* = B_2$, so taking the trace of these equations shows that $y$ and $z$ are unit vectors), so there exists a unitary $U$ such that $UX=Y$.  Regard the first two rows of $U^*$ as $[F_1^* F_2^* F_3^* F_4^*]$.
Then the map
$$X \mapsto F_1XF_1^* + \cdots + F_4XF_4^*$$
is the desired TPCP map.
\qed

\medskip\noindent
\textbf{Remark.} Consider the problem of the existence of a TPCP map
$T$ such that $T(A_i) = B_i$ for $i = 1, \dots, k$,  for given density matrices
$A_1, \dots, A_k, B_1, \dots, B_k \in M_2$.
Evidently, we can focus on the case when $\{A_1, \dots, A_k\}$ is a linearly
independent set. If $k = 1$, then the map defined by $T(X) = (\tr X)B_1$
is a TPCP map satisfying the desired condition.
Theorems \ref{2.1} and \ref{T:qubit} provide conditions for the existence of the desired
TPCP map when $k = 2$.
If $k = 4$, then $\{A_1, \dots, A_4\}$ is a basis for $M_2$. There is a unique
linear map $T$ satisfying $T(A_i) = B_i$ for $i = 1, \dots, 4$.
It is then easy to determine whether $T$ is TPCP by considering its action on the standard
basis $\{E_{11}, E_{12}, E_{21}, E_{22}\}$ for $M_2$.  One simply checks whether
$\tr(E_{11}) = \tr(E_{22}) = 1$, $\tr(E_{12}) = \tr(E_{21}) = 0$, and whether the Choi matrix
$$\begin{bmatrix} T(E_{11}) & T(E_{12}) \cr T(E_{21}) & T(E_{22})\cr\end{bmatrix}$$
is positive semidefinite; see \cite{C}. The remaining case is when $k = 3$.
Again, we can replace $A_1, A_2, A_3$ by suitable linear combinations (and apply the
same linear combinations to $B_1, B_2, B_3$ accordingly) and assume that
$A_i = x_ix_i^*$ for $i = 1, 2, 3$. We have the following result.

\begin{theorem} Suppose $A_i = x_ix_i^*, B_i \in M_2$ are density matrices for $i = 1,2,3$
such that $A_1, A_2, A_3$ are linearly independent.
Let $x_3 = \alpha_1e^{it_1}x_1 + \alpha_2e^{it_2}x_2$ with $\alpha_1, \alpha_2 > 0$,
$t_1, t_2 \in [0, 2\pi)$, and
$$\tilde{B}_3 = \frac{1}{2\alpha_1 \alpha_2}(B_3 - \alpha_1^2 B_1 - \alpha_2^2 B_2).$$
Then there is a TPCP map sending $x_i x_i^*$ to $B_i$ for $i=1,2,3$ if and only if
there exists $C \in M_2$ such that
\begin{equation}\label{eq2}
\tr(CC^*) = 1 + |\det(C)|^2 \le 2, \ \tr \sqrt{B_2} C \sqrt{B_1} = e^{i(t_2-t_1)}x_1^* x_2,
\ \hbox{ and } \ \tilde B_3 = \text{Re} \sqrt{B_2} C \sqrt{B_1}.
\end{equation}
\end{theorem}

\begin{proof} First, consider the forward implication.
Note that  $T$ is a TPCP map sending $x_i x_i^*$ to $B_i$ for $i=1,2$
if and only if $|x_1^* x_2| \leq \|\sqrt{B_1}\sqrt{B_2}\|_1$.
If we write $T(X) = \sum_{j=1}^r F_j X F_j^*$ and $F_j x_i = y_{ij}$, note that $Y_i = \begin{bmatrix} y_{i1} & \dots & y_{ir} \end{bmatrix}$ must equal $\sqrt{B_i} W_i^*$ for some isometry $W_i \in M_{rm}$. Writing $\text{Re } A = (A+A^*)/2$, we have
\begin{align*}
T(x_3 x_3^*) &= \sum_{j=1}^r F_j x_3 x_3^* F_j^* = \sum_{j=1}^r \left( \alpha_1^2 y_{1j} y_{1j}^* + \alpha_2^2 y_{2j} y_{2j}^* + 2 \text{Re} \alpha_1 \alpha_2 e^{i(t_2-t_1)} y_{2j} y_{1j}^* \right) \\
&= \alpha_1^2 B_1 + \alpha_2^2 B_2 + 2\alpha_1 \alpha_2 \text{Re } e^{i(t_2-t_1)}Y_2 Y_1^* \\
&= \alpha_1^2 B_1 + \alpha_2^2 B_2 + 2\alpha_1 \alpha_2 \text{Re} \sqrt{B_2} C \sqrt{B_1},
\end{align*}
where $C$ is a contraction and $\tr \sqrt{B_2} C \sqrt{B_1} = e^{i(t_2-t_1)}x_1^* x_2$.
Note that $C$ is a contraction if and only if the largest eigenvalue
of $CC^*$ is bounded by 1, which is equivalent to the inequalities:
$$\tr(CC^*) \le 1 + \det(CC^*) = 1 + |\det(C)|^2 \le 2.$$
Suppose the first inequality is a strict inequality.
Consider the subspace
$$\cS = \{X \in M_2: \text{Re} \sqrt{B_2} X \sqrt{B_1} = 0, \tr (\sqrt{B_2}X\sqrt{B_1}) = 0\}
\subseteq M_2.$$
Then we may replace $C$ by $C+X$ with $X \in \cS$ so that
$\|C+X\| = 1$, and the new solution $C$ will satisfy
the equality $\tr(CC^*) = 1 + \det(CC^*)$.

Conversely, suppose there exists $C$ satisfying condition (3).
Write $\sqrt{B_1} = [y_{11}\, y_{12}]$, $\sqrt{B_2}C = [y_{21} \, y_{22}]$, and
$\sqrt{B_2}\sqrt{(I-CC^*)} = [y_{23} \, y_{24}]$.  Then the inner product of the two unit vectors $e^{it_1}x_1$ and $e^{it_2} x_2$ equals that of the unit vectors  $\begin{bmatrix} y_{11}\cr y_{12}\cr 0 \cr 0  \end{bmatrix}$ and $\begin{bmatrix} y_{21} \cr y_{22} \cr y_{23} \cr y_{24} \end{bmatrix}$.  Thus, there is a unitary $U \in M_8$ such that
$$U\begin{bmatrix}e^{it_1}x_1 & e^{it_2}x_2\cr 0_6 & 0_6\cr\end{bmatrix} =
\begin{bmatrix} y_{11} & y_{21} \cr y_{12} & y_{22} \cr 0_2 & y_{23} \cr 0_2 & y_{24}\cr\end{bmatrix}.$$
Let the first two rows of $U^*$ be $[F_1^* \, F_2^* \, F_3^* \, F_4^*]$.
Then the map $T(X) = \sum_{j=1}^4 F_jXF_j^*$ satisfies
$T(A_i) = B_i$ for $i = 1,2,3$.
\end{proof}

Note that condition \eqref{eq2} can be verified with standard software. In fact, if we treat C as an unknown
matrix with 4 complex variables (that is, 8 real variables), then the last two equations translate to 5 independent real linear equations.
By elementary linear algebra, the solution has the form $C = C_0 + x_1C_1 + x_2C_2 + x_3 C_3$
for 4 complex matrices $C_0, C_1, C_2, C_3$ in $M_2$, and 3 real variables $x_1, x_2, x_3$. Then we can substitute this expression
into the first equation to see whether the first nonlinear equation (of degree two) is solvable.    In fact, we can formulate
the first equation as an inequality: $\tr(CC^*) \le 1 + |\det(C)|^2 \le 2$.  Then standard computer optimization packages can decide whether there exist real numbers $x_1, x_2, x_3$ satisfying the inequalities.

\section{General states to general states}
\subsection{Moving beyond qubits}

A natural question is whether or not Theorem \ref{T:qubit} can be generalized to non-qubit states, i.e. states on $M_n$ where $n > 2$.
The equivalence of (a) and (b) in Theorem \ref{T:qubit} does not hold for density matrices with dimension greater than two (a counter-example may be found in \cite{HJRW}).  On the other hand, it is known (see \cite[Lemma 1]{CJW}, for example) that the equivalence of (a) and (c) holds for density matrices of any dimension---provided the initial states $A_1$, $A_2$ are pure, i.e. have rank one.  (See the example below.)  This illustrates two points.  First, results for states of arbitrary dimension appear to be more readily attainable when the inputs are restricted to be pure.
Second, this shows why the situation is easier for qubit states: for qubits, one can always perform the reduction described before Theorem \ref{T:qubit} to reduce to the case where the input states are pure, whereas this cannot be done in general for non-qubit states.


\medskip\noindent
\textbf{Example.}
Note that $\|\sqrt {A_1} \sqrt{A_2}\|_1
\le \|\sqrt {B_1} \sqrt{B_2}\|_1$ does not imply
$\|A_1 - tA_2\|_1 \ge \|B_1 -tB_1\|_1$ for all $t \ge 0$
if $A_1$ and $A_2$ are not of rank one.
For example,  let $A_1=\diag(4/5,1/5),A_2=\diag(1/3,2/3)$ and
 $$B_1=\left[\begin{array}{cc}
       1/4&\sqrt{3}/4\\
       \sqrt{3}/4&3/4\end{array}\right],~~~~
       B_2=\left[\begin{array}{cc}
       1/2&1/2\\
       1/2&1/2\end{array}\right].$$
 Then $$\|\sqrt{A_1}\sqrt{A_2}\|_1=0.8815<0.9659=\|\sqrt{B_1}\sqrt{B_2}\|_1$$
 while $$\|A_1-5A_2\|_1=4   <4.1641=\|B_1-5B_2\|_1.$$

So, what more can be said if we impose the additional restriction that the initial states are pure?
Well, if we also assume that the final states are pure, we have the following interesting result from \cite[Theorem 7]{CJW}.

\begin{theorem}\label{T:Chefles}
Let $x_i \in \IC^n$ and $y_i \in \IC^m$ be unit vectors for $i = 1, \dots, k$.  Let $X = [x_1|\dots|x_k]$ and
$Y=[y_1|\dots|y_k]$.  Then there exists a TPCP map $T$ such that $T(x_i x_i^*)=y_i y_i^*$, $i = 1, \dots, k$ if and only if $X^* X = M \circ Y^* Y$ for some correlation matrix $M \in M_k$.
\end{theorem}

Note this gives a computationally efficient condition to check if the matrix $Y^* Y$ has no zero entries.  We will use this result as a model to generalize in the rest of the paper, considering the most general situation first in the next subsection (where we obtain a result which allows us to derive the above theorem as a special case), and then, in the subsequent subsection, we consider keeping pure input states, but relax the condition that the final states be pure.  The final section examines how this theorem changes when the maps are not necessarily trace-preserving.

\subsection{Mixed states to mixed states}

In this subsection we consider the difficult problem of characterizing TPCP maps sending $k$ initial states to $k$ final states (not necessarily of the same dimension), starting with the general case, and then considering special cases that are more tractable.
The following theorem is rather technical, but it does provide a useful framework for the most general situation, and can be readily applied to quickly derive existing results under more specialized circumstances.
The multiple equivalent conditions reflect various approaches and serve as a segue between different viewpoints and lines of attack on a problem.
Note that we ignore zero eigenvalues when using the spectral decomposition in the theorem's statement so as to eliminate redundancies, thus preventing matrices from becoming artificially large.


\begin{theorem}
Suppose $A_1, \dots, A_k \in M_n$ and $B_1, \dots, B_k \in M_m$ are
density matrices.  Using the spectral decomposition, for each $i = 1, \dots, k$, we may write $A_i = X_iD_i^2X_i^*$ and
$B_i = Y_i\tilde D_i^2Y_i^*$, where $X_i, Y_i$ are partial isometries, and $D_i \in M_{r_i}$, $\tilde D_i \in M_{s_i}$
are diagonal matrices whose diagonal entries are given by the square roots of the positive eigenvalues of $A_i$, $B_i$ respectively.
The following conditions are equivalent.

\medskip
{\rm (a)} There is a TPCP map $T:M_n \rightarrow M_m$
such that $T(A_i) = B_i$ for $i = 1, \dots, k$.

{\rm (b)} For each $i = 1, \dots, k$ and $j = 1, \dots, r_i$, there are $s_i \times s$ matrices $V_{ij}$ such that
$$\sum_{j=1}^{r_i} V_{ij}V_{ij}^* = I_{s_i}$$
and the $(p,q)$ entry of the $r_i\times r_j$ matrix
$(D_iX_i^*X_j D_j)$ equals  $\tr(V_{ip}^*\tilde D_i^*Y_i^* Y_j\tilde D_jV_{jq})$.

{\rm (c)}
There are vectors $x_i = \begin{bmatrix} x_{1i} \cr \vdots \cr x_{ri} \end{bmatrix} \in (\IC^n)^r$
and vectors $y_{ji} = \begin{bmatrix} y_{ji}^1 \cr \vdots \cr y_{ji}^s \end{bmatrix} \in (\IC^m)^s$
for $i = 1, \dots, k$ and $j=1, \dots, r$ such that $A_i = \sum_{j=1}^r x_{ji}x_{ji}^*$,
$B_i = \sum_{j=1}^r \sum_{t=1}^s y_{ji}^t (y_{ji}^t)^*$ and
there is a unitary $U \in M_{ms}$ satisfying
$$U \begin{bmatrix} x_{11}& \cdots &x_{r1}& x_{12}&\cdots & x_{rk}\cr 0_{ms-n} & \cdots& 0_{ms-n} & 0_{ms-n} & \cdots & 0_{ms-n} \end{bmatrix}
= \begin{bmatrix}y_{11}^1&\cdots & y_{rk}^1\cr
    \vdots&\vdots\ \vdots \  \vdots &\vdots\cr
       y_{11}^s& \cdots & y_{rk}^s \end{bmatrix}.$$
\end{theorem}

\begin{proof}
(a) $\Rightarrow$ (b).  Let $e_i$ denote the vector with 1 in the $i$th position and 0 in the other positions.  Note that $AA^* \leq BB^*$ in the Loewner order (that is, $BB^* - AA^*$ is positive semidefinite) if and only if $A=BC$ for some contraction $C$.  As in equation (\ref{E:cp}), we may use the operator sum representation for a TPCP map to write $T(A_i) = \sum_{l=1}^s F_l A_i F_l^*$ for some $m \times n$ matrices $F_l$.  Thus
$$(Y_i \tilde{D}_i) (Y_i \tilde{D}_i)^* = B_i = T(A_i) = \sum_{l=1}^s F_l X_i D_i^2 X_i^* F_l^* \geq (F_l X_i D_i e_j) (F_l X_i D_i e_j)^*$$
for any $i,j,l$, whence $F_l X_i D_i e_j = Y_i \tilde{D}_i c_{ij}^l$ for some vectors $c_{ij}^l \in \IC^{s_i}$. Let $V_{ij} = [c_{ij}^1 | \dots | c_{ij}^s]$.  Since $T(A_i) = B_i$ it follows that $\sum_{j=1}^{r_i} V_{ij} V_{ij}^* = I_{s_i}$.

The trace-preserving condition $\sum_{l=1}^s F_l^*F_l = I_n$ implies that there is a unitary matrix $U \in M_{ms}$ whose first $n$ columns are given by $\begin{bmatrix} F_1^* & \dots & F_s^* \end{bmatrix}^*$.  The rest of (b) follows by noting that the inner product of any two columns of the $ms \times (r_1 + \dots + r_k)$ matrix
$$X=\begin{bmatrix} X_1 D_1 & \dots & X_k D_k \\ 0_{ms-n} & \dots & 0_{ms-n} \end{bmatrix}$$
must equal the inner product of the corresponding two columns of $UX$.

\medskip
(b) $\Rightarrow$ (c). Let $r=\max_i r_i$.  Set $x_{ji} = X_i D_i e_j$ if $j \leq r_i$ and $x_{ji} = 0$ otherwise. Let $y_{ji}^l = Y_i \tilde{D}_i V_{ij} e_l$ for $l=1, \dots, s$ if $j \leq r_i$, and set $y_{ij}^l = 0$ otherwise. Then the summations to $A_i$ and $B_i$ are clearly satisfied.  Finally, the last condition of (b) implies that the inner product of $x_{pi}$ and $x_{qj}$ equals the inner product of $y_{pi}$ and $y_{qj}$, and this entails the existence of a unitary $U$ satisfying the final condition of (c).

\medskip
(c) $\Rightarrow$ (a).  Let $\begin{bmatrix} F_1^* & \dots & F_s^* \end{bmatrix}$ be the first $n$ rows of $U^*$, and define $T$ by $T(X) = \sum_{j=1}^s F_j X F_j^*$.  The result follows.
\end{proof}

The conditions (b) and (c) are not easy to check. It would be interesting to find more explicit
and computationally efficient conditions. Nonetheless,
one can use the above theorem to deduce Corollary 10 in \cite{CJW} for
TPCP maps from general states to pure states.

\begin{corollary}
Suppose $A_1, \dots, A_k \in M_n$ and $B_1 = yy_1^*, \dots, B_k = y_ky_k^* \in M_m$ are
density matrices.  For each $i = 1, \dots, k$, write $A_i = X_iD_i^2X_i^*$
such that $D_i \in M_{r_i}$ are diagonal matrices with positive diagonal entries.
The following conditions are equivalent.

\medskip
{\rm (a)} There is a TPCP map $T:M_n \rightarrow M_m$
such that $T(A_i) = B_i$ for $i = 1, \dots, k$.

{\rm (b)} For each $i = 1, \dots, k$ and $j = 1, \dots, r_i$, there are
vectors $v_{ij}$ such that $\sum_{j=1}^{r_i} v_{ij}^*v_{ij} = 1$ and
the $(p,q)$ entry of the $r_i\times r_j$ matrix
$(D_iX_i^*X_j D_j)$ equals  $v_{ip}^* v_{jq} y_i^* y_j$.

{\rm (c)}
For each $i = 1, \dots, k$ and $j = 1, \dots, r_i$, there are
vectors $v_{ij}$ such that $\sum_{j=1}^{r_i} v_{ij}^*v_{ij} = 1$ and a unitary $U$ satisfying
$$U  \begin{bmatrix} X_1D_1 & \cdots & X_kD_k \cr 0 & \dots & 0 \cr\end{bmatrix}
= \begin{bmatrix}v_{11} \otimes y_1 & \cdots & v_{1r_1}\otimes y_1 & \cdots
& v_{k1} \otimes y_k & \cdots & v_{kr_k}\otimes & y_k\end{bmatrix}.$$
\end{corollary}


When all $A_i$ and $B_i$ are of rank one, the above result reduces to the following
result.

\begin{corollary} \label{3.3}
Suppose $x_1, \dots, x_k\in \IC^n$ and $y_1, \dots, y_k \in \IC^m$
are unit vectors.
The following conditions are equivalent.

\medskip
{\rm (a)} There is a TPCP map $T:M_n \rightarrow M_m$
such that $T(x_ix_i^*) = y_iy_i^*$ for $i = 1, \dots, k$.

{\rm (b)} There exist
$D_p = \diag(v_{1p}, \dots, v_{kp})$ for $p = 1, \dots, s$
satisfying
$$\sum_{j=1}^s D_j^*D_j = I_k \quad \hbox{ and }
\quad (x_i^*x_j) =  \sum_{p=1}^s D_p^*(y_i^*y_j)D_p.$$

{\rm (c)} There is a correlation matrix $C \in M_k$ such that
$$(x_i^* x_j) = C \circ (y_i^* y_j).$$

{\rm (d)}
There are unit vectors $v_1, \dots, v_k \in \IC^s$ and a unitary $U\in M_{ms}$ such that
$$U \begin{bmatrix} x_{1}& \cdots &x_{k}\cr 0_{ms-n} & \cdots & 0_{ms-n} \end{bmatrix}
= \begin{bmatrix}v_1 \otimes y_1&\cdots & v_k \otimes y_k \cr\end{bmatrix}.$$
\end{corollary}

Note the equivalence of conditions (a) and (c) above is just Theorem \ref{T:Chefles}.

\subsection{Pure states to mixed states}

Next, we turn to TPCP maps sending pure states to possibly mixed states,
and give a number of necessary and sufficient conditions for their existence.
This problem was also considered in \cite{CJW} using the concept of multi-probabilistic
transformations. We instead rely on purifications of mixed states,
with the aim of generalizing Theorem \ref{T:Chefles}.

\medskip
\begin{theorem} \label{T:purify}
Suppose $x_1, \dots, x_k \in \IC^n$ are unit vectors and $B_1, \dots, B_k \in M_m$ are
density matrices.
Then there is a TPCP map $T$ such that $T(x_ix_i^*) = B_i$ for $i = 1, \dots, k$ if and only if there exist purifications $y_i$ of $B_i$ such that $X^* X = Y^* Y$, where $X=[x_1 | \dots | x_k]$ and $Y = [y_1 | \dots | y_k]$.

\end{theorem}

\begin{proof}
Suppose there is a TPCP map $T$ such that $T(x_i x_i^*) = B_i$.  Write $T(A) = \sum_{j=1}^r F_j A F_j^*$. Since $T$ is trace-preserving, $\begin{bmatrix} F_1^* & F_2^* & \dots & F_r^* \end{bmatrix}$ has orthonormal rows and can be extended to a unitary matrix $U^* \in M_{mr}$.  Define $y_i = U(x_i \oplus 0_{mr-n})$.
Write
\begin{equation}\label{E:purification}
y_i = \begin{bmatrix} y_{1i} \cr \vdots \cr y_{ri} \end{bmatrix} \in (\IC^m)^r, \quad \tilde{X} = \begin{bmatrix} x_1 & \cdots & x_k\cr 0_{mr-n} & \cdots & 0_{mr-n} \end{bmatrix}.
\end{equation}
Then $F_j x_i = y_{ji}$, and $B_i = T(x_i x_i^*) = \sum_{j=1}^r y_{ji}y_{ji}^*$, so $y_i$ is a purification of $B_i$.  Moreover $Y^* Y = (U \tilde{X})^* U  \tilde{X} = \tilde{X}^* \tilde{X} = X^* X$ as desired.
\medskip

Conversely, suppose we have purifications $y_i$ of $B_i$, written as in (\ref{E:purification}) with $B_i = \sum_{j=1}^r y_{ji}y_{ji}^*$.  If $Y^* Y = X^* X = \tilde{X}^* \tilde{X}$, then, since $Y$ and $\tilde{X}$ have the same dimensions, there exists a unitary $U$ such that $Y = U \tilde{X}$.  Write $U =
\begin{bmatrix} F_1 & * \\ \vdots & * \\ F_r & * \end{bmatrix}$ where each $F_i \in M_{mn}$.  Then the map $T$ defined by $T(A) = \sum_{j=1}^r F_j A F_j^*$ is a TPCP map sending $x_i x_i^*$ to $B_i$ for all $i$.
\end{proof}

\medskip\noindent
\textbf{Remark.} To make the similarity to Theorem \ref{T:Chefles} more apparent, note that both conditions in this theorem are equivalent to
\begin{equation}\label{E:Mcond}
X^* X = M \circ Y^* Y \; \text{ for some correlation matrix } M.
\end{equation}
Indeed, if (\ref{E:Mcond}) holds, write $M = C^* C$ and $C = [c_1 | \dots | c_k]$. Since $M_{ii} = 1$, $c_i$ is a unit vector. Let $\tilde{y}_i = c_i \otimes y_i$ and $\tilde{Y} = [\tilde{y}_1 | \dots | \tilde{y}_k]$.  Then $\tilde{y}_i$, $i=1,\dots,k$, are purifications of $B_i$ and $\tilde{Y}^* \tilde{Y} = X^* X$, so we have the second condition in the theorem. The reverse implication is trivial.
\medskip

One definition for the \emph{fidelity} between two states $A$ and $B$ is
$$F(A,B) = \|\sqrt{A} \sqrt{B}\|_1 = \sup \{|\tr \sqrt{A} \sqrt{B} V| : V \text{ is a contraction} \}.$$
It is known that a necessary (but not in general sufficient) condition for the existence of a TPCP map sending $A_1, \dots, A_k$ to $B_1, \dots, B_k$ is that \begin{equation}\label{E:fidelity}
F(B_i,B_j) \geq F(A_i,A_j) \text{ for all } 1 \leq i,j \leq k
\end{equation}
(see \cite[Lemma 1]{CJW}).  The corollary below allows us to deduce this fact immediately when the input states are pure (since $F(x_i x_i^*, x_j x_j^*) = |x_i^* x_j|$).  It also illustrates what missing information (namely, the partial isometries $V_i$) is needed in conjunction with \eqref{E:fidelity} to create a sufficient condition for the existence of a TPCP map.  Unfortunately, it is still not very computationally efficient.

\begin{corollary}
Suppose $x_1, \dots, x_k \in \IC^n$ are unit vectors and $B_1, \dots, B_k \in M_m$ are
density matrices.
Then there is a TPCP map $T$ such that $T(x_ix_i^*) = B_i$ for $i = 1, \dots, k$ if and only if there exist partial isometries $V_i \in M_{mr}$ such that \begin{equation}\label{E:isom}
\sqrt{B_i} V_i V_i^* \sqrt{B_i} = B_i \quad \text{ and } \quad x_i^* x_j = \tr \sqrt{B_i} \sqrt{B_j} V_j V_i^* \; \text{ for all } i,j.
\end{equation}
\end{corollary}

\begin{proof}
Suppose $V_1, \dots, V_k$ are partial isometries satisfying \eqref{E:isom}.  Let $Y_i = \sqrt{B_i} V_i \in M_{mr}$, write $Y_i = [y_{1i} | \dots | y_{ri}]$, and define $y_i \in \IC^{rm}$ as in \eqref{E:purification}.  Then $B_i = Y_i Y_i^* = \sum_{j=1}^r y_{ji} y_{ji}^*$, so $y_i$ is a purification of $B_i$.  Since $X^* X = Y^* Y$ for $X=[x_1|\dots|x_k]$ and $Y=[y_1|\dots|y_k]$, the result follows by Theorem \ref{T:purify}.

Conversely, by Theorem \ref{T:purify}, we may assume there are purifications $y_i$ of $B_i$ in the form of \eqref{E:purification} and $X^* X = Y^* Y$.  Let $Y_i = [y_{1i}| \dots | y_{ri}] \in M_{mr}$.  Since $Y_i Y_i^* = B_i$, there exist partial isometries $V_i \in M_{mr}$ such that $Y_i = \sqrt{B_i} V_i$.  But $x_i^* x_j = \tr Y_i^* Y_j$, so \eqref{E:isom} holds.
\end{proof}

\section{General physical transformations on pure states}

Theorem \ref{T:Chefles} (quoted from \cite{CJW}) gives
a simple criterion for the existence of a TPCP map sending pure states
$x_1 x_1^*, \dots, x_k x_k^*$ to pure states $y_1 y_1^*, \dots, y_k y_k^*$.
One might wonder how to generalize this criterion to handle arbitrary interpolating completely positive (CP) maps.
The remark in \cite{CJW} after Theorem 7 seems to assert that there exists a CP
map sending $x_i x_i^*$ to $y_i y_i^*$ if and only if $X^*X = M \circ Y^* Y$ for some positive semidefinite $M$ (without any restriction on the diagonal entries of $M$).  However, this condition
is neither necessary nor sufficient.

For example, let $\{e_1, e_2\}$ be the standard basis for $\IC^2$.
Take $x_1 = e_1$, $x_2 = (e_1 + e_2)/\sqrt{2}$ and $y_1 = e_1$, $y_2 = e_2$.
Then $M \circ Y^*Y = M \circ I$ is diagonal for any matrix $M$, but $X^* X$
has nonzero off-diagonal entries, so the condition is not satisfied.
Nonetheless, there is a CP map sending $x_i x_i^*$ to $y_i y_i^*$;
let $S \in M_2$ be such that $Sx_i = y_i$.  Then the CP map $T(A) = SAS^*$ works.

On the other hand, let $x_1 = x_2 = e_1$.  Let $y_1 = e_1$, $y_2 = 2e_1$.
Let $M = (e_1 + 0.5 e_2)(e_1 + 0.5 e_2)^*$.  Then $X^* X = M \circ Y^*Y$ is the matrix of all ones.
But clearly there is no map $T$ sending $e_1 e_1^*$ to both $e_1 e_1^*$ and $4 e_1 e_1^*$.

The following results consider interpolating CP maps and unital CP maps, generalizing Theorem \ref{T:Chefles}, and giving necessary and sufficient conditions in the same spirit as \cite{CJW}.

\begin{theorem}\label{P:cp}
Fix positive semi-definite rank-one matrices $x_i x_i^* \in M_n$ and $y_i y_i^* \in M_m$ for
$i=1 \dots k$.  Let $X = [x_1 | x_2 | \dots | x_k]$ and $Y=[y_1 | y_2 | \dots | y_k]$. Then there
exists a completely positive map $T$ such that $T(x_i x_i^*) = y_i y_i^*$ if and only if there
exists a positive semi-definite matrix $M \in M_k$ with
$M_{ii}=1$ such that $\ker X^*X \subseteq \ker M \circ (Y^*Y)$.
\end{theorem}

\begin{proof}
There exists a completely positive map $T$ such that $T(x_i x_i^*) = y_i y_i^*$ if and only if
\begin{align*}
&\exists F_1, \dots, F_r \in M_{mn} \text{ such that } \sum_{j=1}^r F_j x_i x_i^* F_j^* = y_i y_i^* \quad \forall i = 1, \dots, k\\
\iff &\exists F_1, \dots, F_r \in M_{mn} \text{ and unit vectors } c_1, \dots, c_k \in \IC^r \text{ such that } F_j x_i = c_{ij} y_i \\
\iff &\exists F_j \in M_{mn}, \text{unit vectors } c_i \in \IC^r \text{ so that }
F_j X = Y \Gamma_j
\text{ where } \Gamma_j \text{ is diagonal with } (\Gamma_j)_{ii} = c_{ij}\\
\iff &\exists \text{ diagonal } \Gamma_j \in M_k \text{ with } \sum_{j=1}^r \Gamma_j \Gamma_j^* = I_k \text{ such that } \row Y \Gamma_j \subseteq \row X \; \forall j\\
&(\text{equivalently, } \ker X \subseteq \ker Y\Gamma_j \; \forall j, \text{ or } \ker X^*X \subseteq \ker \Gamma_j^* Y^* Y \Gamma_j \; \forall j)\\
\iff &\ker X^*X \subseteq \ker M \circ Y^*Y \text{ where } (M_t)_{ij} = 
(\overline{\Gamma}_t)_{ii} (\Gamma_t)_{jj}\\
&\text{ and } M = \sum_{t=1}^r M_t \text{ is a positive semi-definite matrix with } M_{ii} = 1
\end{align*}
\end{proof}

We will present a result on
unital completely positive maps sending pure states to pure states as a corollary
of the following more general result.
Recall that for a rank $r$ positive semi-definite matrix
$A \in M_n$ with spectral decomposition $A = \lambda_1 u_1 u_1^* + \cdots + \lambda_r u_r u_r^*$,
where $\{u_1, \dots, u_k\} \subset \IC^n$ is an orthonormal set of eigenvectors of $A$
corresponding to the positive eigenvalues $\lambda_1, \dots, \lambda_r$,
the Moore-Penrose inverse $A^+$ of $A$ has the spectral decomposition
$A^+ = \lambda_1^{-1} u_1 u_1^* + \cdots + \lambda_r^{-1} u_r u_r^*$.

\begin{theorem}\label{P:unital}
Fix rank-one matrices $x_i x_i^* \in M_n$ and $y_i y_i^* \in M_m$ for $i=1 \dots k$.  Fix a positive semi-definite matrix $B \in M_m$.  Let $X = [x_1 | x_2 | \dots | x_k]$ and $Y=[y_1 | y_2 | \dots | y_k]$. Then there exists a completely positive linear map $T: M_n \rightarrow M_m$ such that
$$T(I) = B \quad \hbox{ and } \quad
T(x_ix_i^*) = y_iy_i^* \quad \hbox{ for } i = 1, \dots, k,$$
if and only if there exists a positive semi-definite matrix $M \in M_k$ with $M_{ii} = 1$ such that
$$(1) \; \ker X^*X \subseteq \ker M \circ (Y^*Y) \quad \text{ and } \quad (2) \; Y[\bar{M} \circ (X^*X)^{+}]Y^* \leq B,$$
(with equality in (2) should $X$ have rank $n$).  Here $X^{+}$ denotes the Moore-Penrose inverse of $X$.
\end{theorem}

\begin{proof}
Note that the existence of a CP map $T$ such that $T(I)=B$ and $T(x_i x_i^*) = y_i y_i^*$ is equivalent to the existence of $F_1, \dots, F_r \in M_{mn}$ satisfying
$$(a) \; \sum_{j=1}^r F_j x_i x_i^* F_j^* = y_i y_i^* \quad \forall i \quad \text{ and } \quad (b) \; \sum_{j=1}^r F_j F_j^* = B.$$

Proof of Necessity: Condition (a) and the proof of Theorem \ref{P:cp} imply that $F_j X = Y \Gamma_j$ for some diagonal $\Gamma_j \in M_k$ with $\sum_{j=1}^r \Gamma_j \Gamma_j^* = I_k$.  Moreover condition (1) follows with the matrix $M$ defined by
$M_{ij} = \sum_{t=1}^r (\overline{\Gamma}_t)_{ii} (\Gamma_t)_{jj}$.
\medskip

Let $P$ denote the orthogonal projection $XX^{+}$, and let $P^{\perp} = I_n-P$.  Then $F_jP = F_j XX^{+} = Y \Gamma_j X^{+}$, so
\begin{align*}
B &= \sum_{j=1}^r F_j F_j^* = \sum_{j=1}^r (F_j P + F_j P^{\perp})(P F_j^* + P^{\perp} F_j^*) = \sum_{j=1}^r F_j P P F_j^* + F_j P^{\perp} F_j^* \\
&= \sum_{j=1}^r Y \Gamma_j X^{+} (X^{+})^* \Gamma_j^* Y^* + \sum_{j=1}^r F_j P^{\perp} F_j^* \qquad \text{ but } X^{+} (X^{+})^* = (X^*X)^{+}\\
&= Y [\bar{M} \circ (X^* X)^{+}] Y^* + \sum_{j=1}^r F_j P^{\perp} F_j^* \\
&\geq Y [\bar{M} \circ (X^* X)^{+}] Y^*
\end{align*}
with equality if $P=I_n$, that is, if $X$ has rank $n$.

\bigskip
Proof of Sufficiency: Since $M$ is positive semi-definite with $M_{ii}=1$, we can write $M=C^*C$ where $C=[c_1|c_2|\dots|c_k] \in M_{rk}$, and $c_i$ is a unit vector for all $i$.  If necessary, we may append extra zero entries to the end of each $c_i$ so that we may assume $r \geq m$.  Define diagonal matrices $\Gamma_t \in M_k$ by $(\Gamma_t)_{ii} = c_{it}$.   Then
$$M \circ Y^* Y = \sum_{j=1}^r \Gamma_j^* Y^* Y \Gamma_j, \qquad \bar{M} \circ (X^* X)^{+} = \sum_{j=1}^r \Gamma_j (X^*X)^{+} \Gamma_j^*, \quad \text{ and }
\quad \sum_{j=1}^r \Gamma_j \Gamma_j^* = I_k.$$

Condition (2) implies
\begin{align*}
B &= Y[\bar{M} \circ (X^*X)^{+}]Y^* + EE^* \qquad \text{ for some } E\\
&= \sum_{j=1}^r Y \Gamma_j (X^*X)^{+} \Gamma_j^* Y^* + \sum_{j=1}^r G_j P^{\perp} G_j^*
\end{align*}
where we may choose $G_j \in M_{mk}$ so that $G_j P^{\perp} G_j^*$ is proportional to an eigenprojection for $EE^*$ with rank at most one.  Note that $P^{\perp} = 0$ if and only if $X$ has rank $n$.

\medskip
Define $F_j = Y \Gamma_j X^{+} + G_j P^{\perp}$.  Then
\begin{align*}
\sum_{j=1}^r F_j F_j^* &= \sum_{j=1}^r Y \Gamma_j X^{+} (X^{+})^* \Gamma_j^* Y^* + G_j P^{\perp} G_j^* + Y \Gamma_j X^{+} P^{\perp} G_j^* + G_j P^{\perp} (X^{+})^* \Gamma_j^* Y^* \\
&= \sum_{j=1}^r Y \Gamma_j X^{+} (X^{+})^* \Gamma_j^* Y^* + G_j P^{\perp} G_j^* = B
\end{align*}
since $X^{+} P^{\perp} = X^{+} (I-XX^{+}) = 0$,
and the fourth term in the second sum is the adjoint of the third term.
\medskip

On the other hand
\begin{align*}
F_j X &= Y \Gamma_j (X^{+} X - I + I) + G_j P^{\perp} X \\
&= - Y \Gamma_j (I-X^{+} X) + Y \Gamma_j + G_j (I-XX^{+}) X.
\end{align*}
But $I-X^{+} X$ is the orthogonal projection onto $\ker X$; since condition (1) implies $\ker X \subseteq \ker Y \Gamma_j$ for all $j$, the first term must be 0.  And $(I-XX^{+})X = X - XX^{+}X = 0$, so the third term vanishes.  Thus $F_j X = Y \Gamma_j$ for all $j$, and
$$\sum_{j=1}^r F_j x_i x_i^* F_j^* = y_i y_i^* \quad \text{ for all } i = 1, \dots, k$$
as desired.
\end{proof}

\begin{corollary} \label{C:unital}
Fix $x_i x_i^* \in M_n$ and $y_i y_i^* \in M_m$ for $i=1, \dots, k$.  Write $X=[x_1|\dots|x_k]$ and $Y=[y_1|\dots|y_k]$.
Then there exists a unital completely positive map $T$ satisfying $T(x_i x_i^*) = y_i y_i^*$ for all $i = 1, \dots, k$
if and only if there exists a positive semi-definite matrix $M \in M_k$ with $M_{ii} = 1$ such that
$$(1) \; \ker X^*X \subseteq \ker M \circ (Y^*Y) \quad \text{ and } \quad (2) \; Y[\bar{M} \circ (X^*X)^{+}]Y^* \leq I_m,$$
(with equality in (2) should $X$ have rank $n$).
\end{corollary}

\begin{proof}
Take $B=I_m$ in Theorem \ref{P:unital}.
\end{proof}

\begin{corollary} Use the notation in Corollary \ref{C:unital}.
There is a unital TPCP map sending $x_1 x_1^*, \dots, x_k x_k^*$ to $y_1 y_1^*, \dots, y_k y_k^*$ if and only if
$m = n$ and
there exists a positive semi-definite matrix $M \in M_k$ with $M_{ii} = 1$ such that
$$(1) \; X^*X = M \circ (Y^*Y) \quad \text{ and }
\quad (2) \; Y[\bar{M} \circ (X^*X)^{+}]Y^* \leq I_n,$$
with equality in (2) should $X$ have rank $n$.
\end{corollary}

\medskip\noindent
\textbf{Note}

Reference \cite{HJRW} was brought to our attention by one of the referees.  In it the authors independently obtain our condition (c) in Theorem \ref{T:qubit}.  Moreover, they extend the result by allowing final states to have dimension greater than two, although it appears that our proof is self-contained, and uses more elementary methods.  They also consider the problem of \emph{approximately} mapping a set of initial states to a set of final states via CP maps.

\medskip\noindent
{\bf Acknowledgments}

This research was supported by an RCG grant with Sze as PI and Li as co-PI.
The grant supported the post-doctoral fellowship of Huang at the Hong Kong Polytechnic
University, and the visit of Poon to the University of Hong Kong and Hong Kong Polytechnic
University in the spring of 2012. Li was also supported by a USA NSF grant; he
was a visiting professor of the University of Hong Kong in the spring of 2012,
an honorary professor of Taiyuan University of Technology (100 Talent Program scholar),
and an honorary professor of the  Shanghai University.

The authors would like to thank Dr. J. Wu and Dr. L. Zhang
for drawing their attention to the papers \cite{AU} and \cite{CJW}; and thank
Dr. H.F. Chau, Dr. W.S. Cheung, Dr. C.H. Fung, and Dr. Z.D. Wang for helpful discussion.  The authors would also like to thank the referees and editors for their helpful comments to improve this paper.

\medskip\noindent
{\bf Addresses}

\noindent
Zejun Huang and Nung-Sing Sze\\
Department of Applied Mathematics, Hong Kong Polytechnic University, Hung Hom, Hong Kong. \\
huangzejun@yahoo.cn, raymond.sze@inet.polyu.edu.hk

\noindent
Chi-Kwong Li \\
Department of Mathematics, College of William and Mary, Williamsburg, VA 23187, USA. \\
ckli@math.wm.edu

\noindent
Edward Poon \\
Department of Mathematics,
 Embry-Riddle Aeronautical University,
Prescott, AZ 86301, USA.  \\
poon3de@erau.edu


\begin{thebibliography}{WWW}

\bibitem{AU} P. Alberti and A. Uhlmann, ``A problem relating to positive linear maps on matrix algebras,"
Rep. Math. Phys. \textbf{18}, 163 (1980).

\bibitem{CJW} A. Chefles, R. Jozsa, and A. Winter, ``On the existence of physical transformations between sets of quantum states," International J. Quantum Information \textbf{2}, 11-21 (2004).

\bibitem{C} M.D. Choi, ``Completely positive linear maps on complex matrices,"
Linear Algebra Appl 10, 285-290 (1975).

\bibitem{HJRW} Heinosaari, Jivulescu, Reeb, and Wolf, ``Extending quantum operations," arXiv:1205.0641.

\bibitem{K} K. Kraus, ``States,  Effects, and Operations: Fundamental Notions of Quantum Theory,"
Lecture Notes in Physics, Vol. 190. Springer-Verlag, Berlin, 1983.

\bibitem{LP} C.K. Li and Y.T. Poon,
``Interpolation by completely positive maps," Linear and Multilinear Algebra \textbf{59}, 1159-1170 (2011).

\bibitem{NC} M.A. Nielsen and I.L. Chuang, ``Quantum Computation and Quantum
Information," Cambridge University Press, Cambridge, 2000.
\end{thebibliography}
\end{document}